\documentclass[10pt, draftcls, onecolumn]{IEEEtran}
\usepackage{subfigure}
\usepackage{setspace}
\usepackage{amsmath}
\usepackage{amssymb}
\usepackage{amsfonts}
\usepackage{amscd}
\usepackage{mathrsfs}
\usepackage[final]{graphicx}
\usepackage{graphicx}
\usepackage{psfrag}
\usepackage{color}
\usepackage{url}

\newtheorem{theorem}{Theorem}
\newtheorem{lemma}{Lemma}
\newtheorem{definition}{Definition}

\newtheorem{corollary}{Corollary}
\newtheorem{example}{Example}
\newtheorem{remark}{Remark}

\newtheorem{case}{Case}

\title{High Rate Single-Symbol ML Decodable Precoded DSTBCs for Cooperative Networks}
\author{Harshan J and B. Sundar Rajan, Senior Member, IEEE
\thanks{This work was supported through grants to B.S.~Rajan; partly by the DRDO-IISc program on Advanced Research in Mathematical Engineering, and partly by the Council of Scientific \& Industrial Research (CSIR, India) Research Grant (22(0365)/04/EMR-II). Part of the content of this paper has been submitted to IEEE International Conference on Communications (ICC 2008).
The authors are with the Department of Electrical Communication Engineering, Indian Institute of Science, Bangalore-560012, India. Email:\{harshan,bsrajan\}@ece.iisc.ernet.in.}}
\begin{document}
\maketitle
%
\begin{abstract}
Distributed Orthogonal Space-Time Block Codes (DOSTBCs) achieving full diversity order and single-symbol ML decodability have been introduced recently by Yi and Kim for cooperative networks and an upperbound on the maximal rate of such codes along with code constructions has been presented. In this paper, we introduce a new class of Distributed STBCs called Semi-orthogonal Precoded Distributed Single-Symbol Decodable STBCs (S-PDSSDC) wherein, the source performs co-ordinate interleaving of information symbols appropriately before transmitting it to all the relays. It is shown that DOSTBCs are a special case of S-PDSSDCs. A special class of S-PDSSDCs having diagonal covariance matrix at the destination is studied and an upperbound on the maximal rate of such codes is derived. The bounds obtained are approximately twice larger than that of the DOSTBCs. A systematic construction of S-PDSSDCs is presented when the number of relays $K \geq 4$. The constructed codes are shown to achieve the upperbound on the rate when $K$ is of the form 0 or 3 modulo 4. For the rest of the values of $K$, the constructed codes are shown to have rates higher than that of DOSTBCs. It is shown that S-PDSSDCs cannot be constructed with any form of linear processing at the relays when the source doesn't perform co-ordinate interleaving of the information symbols. Simulation result shows that S-PDSSDCs have better probability of error performance than that of DOSTBCs.
\end{abstract}
\begin{keywords}
Cooperative diversity, single-symbol ML decoding, distributed space-time coding, precoding.
\end{keywords}
\section{Introduction and preliminaries}

\indent \textbf{C}ooperative communication has been a promising means of achieving spatial diversity without the need of multiple antennas at the individual nodes in a wireless network. The idea is based on the relay channel model, where a set of distributed antennas belonging to multiple users in the network co-operate to encode the signal transmitted from the source and forward it to the destination so that the required diversity order is achieved, \cite{SEA1}-\cite{NBK}. Spatial diversity obtained from such a co-operation is referred to as co-operative diversity. In \cite{JiH1}, the idea of space-time coding devised for point to point co-located multiple antenna systems is applied for a wireless relay network and is referred to as distributed space-time coding. The technique involves a two phase protocol where, in the first phase, the source broadcasts the information to the relays and in the second phase, the relays linearly process the signals received from the source and forward them to the destination such that the signal at the destination appears as a Space-Time Block Code (STBC).\\
\indent Since the work of \cite{SEA1}-\cite{JiH1}, lot of efforts have been made to generalise the various aspects of space-time coding proposed for multiple antenna systems to the co-operative setup. One such important aspect is the design of low-complexity Maximum Likelihood (ML) decodable Distributed Space-Time Block Codes (DSTBCs) - in particular, the design of Single-Symbol ML Decodable (SSD) DSTBCs. For a background on SSD STBCs for MIMO systems, we refer the reader to \cite{TJC} - \cite{SaR}. Through out the paper, we consider DSTBCs that are ML decodable. Two group decodable DSTBCs were introduced in \cite{KiR} through doubling construction using a commuting set of matrices from field extensions. In \cite{JiJ1}, Orthogonal Designs (ODs) and Quasi-orthogonal Designs \cite{WWX} originally proposed for multiple antenna systems have been applied to the co-operative framework. Since the co-variance matrix of additive noise at the destination is a function of (i) the realisation of the channels from the relays to the destination and (ii) the relay matrices, complex orthogonal designs (except for 2 relays - Alamouti code) loose their SSD property in the co-operative setup. In \cite{RaR1}, DSTBCs based on co-ordinate interleaved orthogonal designs \cite{KhR} have been introduced which have reduced decoding complexity. In this set-up, the source performs co-ordinate interleaving of information symbols before transmitting to the relays. In \cite{RaR2}, low decoding complexity  DSTBCs were proposed using Clifford-algebras, wherein the relay nodes are assumed to have the knowledge of the phase component of the source-to-relay channels. A class of  four-group decodable DSTBCs was also proposed in \cite{RAR}.\\
\indent Recently, in \cite{ZhK}, Distributed Orthogonal Space-Time Codes (DOSTBCs) achieving single-symbol decodability have been introduced for co-operative networks. The authors considered a special class of DOSTBCs which make the covariance matrix of the additive noise vector at the destination, a diagonal one and such a class of codes was referred to as row monomial DOSTBCs. Upperbounds on the maximum symbol-rate (in complex symbols per channel use in the second phase) of row monomial DOSTBCs have been derived and a systematic construction of such codes has been proposed. The constructed codes were shown to meet the upperbound for even number of relays. In \cite{ZhK2}, the same authors have derived an upperbound on the symbol-rate of DOSTBCs when the additive noise at the destination is correlated and have shown that the improvement in the rate is not significant when compared to the case when the noise at the destination is uncorrelated \cite{ZhK}.\\
\indent In \cite{SCS} and \cite{ZhK2}, SSD DSTBCs have been studied when the relay nodes are assumed to know the corresponding channel phase information. An upperbound on the symbol rate for such a set up is shown to be $\frac{1}{2}$ which is independent of the number of relays.\\
\indent In \cite{ZhK}, \cite{SCS} and \cite{ZhK2} the source node transmits the information symbols to all the relays with out any processing. On the similar lines of \cite{RaR1} and using the framework proposed in \cite{ZhK}, in this paper, we propose SSD DSTBCs aided by linear precoding of the information vector at the source. In our set-up, we assume that the relay nodes do not have the knowledge of the channel from the source to itself. In particular, it is shown that, co-ordinate interleaving of information symbols at the source along with the appropriate choice of relay matrices, SSD DSTBCs with maximal rates higher than that of DOSTBCs can be constructed. The contributions of this paper can be summarized as follows:
\begin{itemize}
\item A new class of DSTBCs called Precoded DSTBCs (PDSTBCs) (Definition \ref{def_pdstbc}) is introduced where the source performs co-ordinate interleaving of information symbols appropriately before transmitting it to all the relays. Within this class, we identify codes that are SSD and refer to them as Precoded Distributed Single Symbol Decodable STBCs (PDSSDCs) (Definition \ref{D_pdssd}). The well known DOSTBCs studied in \cite{ZhK} are shown to be a special case of PDSSDCs.
\item A set of necessary and sufficient conditions on the relay matrices for the existence of PDSSDCs is proved (Theorem \ref{necc_suf_cond}).
\item Within the set of PDSSDCs, a class of Semi-orthogonal PDSSDCs (S-PDSSDC) (Definition \ref{D_so_updssd}) is defined. The known DOSTBCs are shown to belong to the class of S-PDSSDCs. On the similar lines of \cite{ZhK}, a special class of S-PDSSDCs having a diagonal covariance matrix at the destination is studied and are referred to as row monomial S-PDSSDCs. An upperbound on the maximal symbol-rate of row monomial S-PDSSDCs is derived. It is shown that, the symbol rate of row monomial S-PDSSDC is upperbounded by $\frac{2}{l}$ and $\frac{2}{l + 1}$, when the number of relays, $K$ is of the form $2l$ and $2l + 1$ respectively, where $l$ is any natural number. The bounds obtained are approximately twice larger than that of DOSTBCs.
\item A systematic construction of row-monomial S-PDSSDCs is presented when $K \geq 4$. Codes achieving the upperbound on the symbol rate are constructed when $K$ is 0 or 3 modulo 4. For the rest of the values of $K$, the constructed S-PDSSDCs are shown to have rates higher than that of the DOSTBCs.
\item Precoding of information symbols at the source has resulted in the construction of high rate S-PDSSDCs. In this setup, the relays do not perform co-ordinate interleaving of the received symbols. It is shown that, when the source transmits information symbols to all the relays with out any precoding, and if the relays are allowed to perform linear processing of their received vector, S-PDSSDCs other than DOSTBCs cannot be constructed thereby, necessitating the source to perform coordinate interleaving of information symbols in order to construct high rate S-PDSSDCs.
\end{itemize}
\indent The remaining part of the paper is organized as follows: In Section \ref{sec2}, along with the signal model,  PDSTBCs are introduced and a special class of it called PDSSDCs is defined. A set of necessary and sufficient conditions on the relay matrices for the existence of PDSSDCs is also derived. In Section \ref{sec3}, S-PDSSDCs are defined and a special class of it called row-monomial S-PDSSDCs are studied. An upperbound on the maximal rate of row-monomial S-PDSSDCs is derived. In Section \ref{sec4}, construction of row-monomial S-PDSSDCs is presented along with some examples. In Section \ref{sec5}, we show that the source has to necessarily perform precoding of information symbols in order to construct high rate S-PDSSDCs. The problem of designing two-dimensional signal sets for the full diversity of RS-PDSSDCs is discussed in Section \ref{sec6} along with some simulation results. Concluding remarks and possible directions for further work constitute Section \ref{sec7}.\\

\textit{Notations:} Through out the paper, boldface letters and capital boldface letters are used to represent vectors and matrices respectively. For a complex matrix $\textbf{X}$, the matrices $\textbf{X}^*$, $\textbf{X}^T$,  $\textbf{X}^{H}$, $|\textbf{X}|$, $\mbox{Re}~\textbf{X}$ and $\mbox{Im}~\textbf{X}$ denote, respectively, the conjugate, transpose, conjugate transpose, determinant, real part and imaginary part of $\textbf{X}$. The element in the $r_1^{th}$ row and the $r_2^{th}$ column of the matrix $\textbf{X}$ is denoted by $[\textbf{X}]_{r_1,r_2}$. The diagonal matrix $\mbox{diag}\left\lbrace [\textbf{X}]_{1,1}, [\textbf{X}]_{2,2} \cdots [\textbf{X}]_{T,T}\right\rbrace$ constructed from the diagonal elements of a $T \times T$ matrix $\textbf{\textbf{X}}$ is denoted by diag$\left[\textbf{X}\right]$. For complex matrices $\textbf{X}$ and $\textbf{Y}$, $\textbf{X} \otimes \textbf{Y}$ denotes the tensor product of $\textbf{X}$ and $\textbf{Y}$. The tensor product of the matrix $\textbf{X}$ with itself $r$ times where $r$ is any positive integer is represented by $\textbf{X}^{\otimes^r}$. The $ T\times T$ identity matrix and the $T \times T$ zero matrix respectively denoted by $\textbf{I}_T$ and $\textbf{0}_T$. The magnitude of a complex number $x$, is denoted by $|x|$ and $E \left[x\right]$ is used to denote the expectation of the random variable $x.$ A circularly symmetric complex Gaussian random vector, $\textbf{x}$ with mean $\mu$ and covariance matrix $\mathbf{\Gamma}$ is denoted by $\textbf{x} \sim \mathcal{CSCG} \left(\mu, \mathbf{\Gamma} \right) $. The set of all integers, the real numbers and the complex numbers are respectively, denoted by ${\mathbb Z}$, $\mathbb{R}$ and ${\mathbb C}$ and $j$ is used to represent $\sqrt{-1}.$ The set of all $T \times T$ complex diagonal matrices is denoted by $\mathcal{D}_{T}$ and a subset of $\mathcal{D}_{T}$ with strictly positive diagonal elements is denoted by $\mathcal{D}_{T}^{+}.$
\section{Precoded distributed space-time coding}\label{sec2}
\subsection{Signal model}
The wireless network considered as shown in Figure \ref{model_network} consists of $K + 2$ nodes each having single antenna which are placed randomly and independently according to some distribution. There is one source node and one destination node. All the other $K$ nodes are relays. We denote the channel from the source node to the $k^{th}$ relay as $h_{k}$ and the channel from the $k^{th}$ relay to the destination node as $g_{k}$ for $k=1,2, \cdots, K$.
\noindent The following assumptions are made in our model:

\begin{itemize}
\item All the nodes are subjected to half duplex constraint.
\item Fading coefficients  $h_{k},g_{k}$ are i.i.d $ \mathcal{CSCG} \left(0,1 \right)$ with coherence time interval of atleast $N$ and $T$ respectively.
\item All the nodes are synchronized at the symbol level.
\item Relay nodes do not have the knowledge of fade coefficients $h_{k}$.
\item Destination knows the fade coefficients $g_{k}$, $h_{k}$.
\end{itemize}
\noindent The source is equipped with a $N$ length complex vector from the codebook $\mathcal{S}$ = $\left\{ \textbf{s}_{1},\, \textbf{s}_{2},\, \textbf{s}_{3},\, \cdots, \textbf{s}_{L} \right\} $ consisting of information vectors $\textbf{s}_{l} \in \mathbb{C}^{1 \times N}$ such that $E\left[\textbf{s}_{l}\textbf{s}_{l}^{H}\right]$ = 1 for all $l= 1,\cdots, L$. The source is also equipped with a pair of $N \times N$ matrices $\textbf{P}$ and $\textbf{Q}$ called precoding matrices. Every transmission from the source to the destination comprises of two phases. When the source needs to transmit an information vector $\textbf{s} \in \mathcal{S}$ to the destination, it generates a new vector $\tilde{\textbf{s}}$ as,
\begin{equation}
\label{cord_int}
\tilde{\textbf{s}} = \textbf{s} \textbf{P} + \textbf{s}^{*}\textbf{Q}
\end{equation}
such that $E\left[\tilde{\textbf{s}}\tilde{\textbf{s}}^{H}\right]$ = 1 and broadcasts the vector $\tilde{\textbf{s}}$ to all the $K$ relays (but not to the destination). The received vector at the $k^{th}$ relay is given by $\textbf{r}_{k} = \sqrt{P_{1}N}h_{k}\tilde{\textbf{s}} + \textbf{n}_{k}$, for all $\textit{ k} = 1,2,\cdots, K$ where $\textbf{n}_{k} \sim \mathcal{CSCG} \left(0,\textbf{I}_{N} \right) $ is the additive noise at the $k^{th}$ relay and $P_{1}$ is the total power used at the source node every channel use. In the second phase, all the relay nodes are scheduled to transmit $T$ length vectors to the destination simultaneously. Each relay is equipped with a fixed pair of $N \times T$ rectangular matrices $\textbf{A}_{k}$, $\textbf{B}_{k}$ and is allowed to linearly process the received vector. The $k^{th}$ relay is scheduled to transmit
\begin{equation}\label{t1j}
\textbf{t}_{k} = \sqrt{\frac{P_{2}T}{(1 + P_{1})N}}\left\lbrace \textbf{r}_{k}\textbf{A}_{k} + \textbf{r}_{k}^{*}\textbf{B}_{k}\right\rbrace .
\end{equation}
where $P_{2}$ is the total power used at each relay every channel use in the second phase. The vector received at the destination is given by
\begin{equation}\label{bfy}
\textbf{y} = \sum_{k = 1}^{K} g_{k}\textbf{t}_{k} + \textbf{w}
\end{equation}
\noindent where $\textbf{w} \sim \mathcal{CSCG} \left(0,\textbf{I}_{T} \right)$ is the additive noise at the destination. Using \eqref{t1j} in \eqref{bfy},  \textbf{y} can be written as
\begin{equation*}
\label{bfy1}
\textbf{y} = \sqrt{\frac{P_{1}P_{2}T}{(1 + P_{1})N}}\textbf{g}\textbf{X} + \textbf{n}
\end{equation*}
\noindent where
\begin{itemize}
\item $\textbf{n} = \sqrt{\frac{P_{2}T}{(1 + P_1)N}}\left[ \sum_{k=1}^{K} g_{k}\left\lbrace \textbf{n}_{k}\textbf{A}_{k} + \textbf{n}_{k}^{*}\textbf{B}_{k}\right\rbrace \right]  + \textbf{w} \label{bfN}.$
\item The equivalent channel \textbf{g} is given by $[g_{1} ~ g_{2} ~ \cdots ~ g_{K} ] \in \mathbb{C}^{1 \times K}.$
\item Every codeword $\textbf{X} \in \mathbb{C}^{K \times T}$ is of the form,
{\small
\begin{equation*}
\textbf{X} = \left[ \left[h_{1}\tilde{\textbf{s}}\textbf{A}_{1} + h_{1}^{*}\tilde{\textbf{s}}^{*}\textbf{B}_{1} \right]^{T} ~~ \left[h_{2}\tilde{\textbf{s}}\textbf{A}_{2} + h_{2}^{*}\tilde{\textbf{s}}^{*}\textbf{B}_{2}\right]^{T} ~~ \cdots ~~ \left[h_{K}\tilde{\textbf{s}}\textbf{A}_{K} + h_{K}^{*}\tilde{\textbf{s}}^{*}\textbf{B}_{K}\right]^{T} \right]^{T}.
\end{equation*}
}
\end{itemize}
\begin{definition}
\label{def_pdstbc}
\noindent The collection $\mathcal{C}$ of $K \times T$ codeword matrices shown below, where $\textbf{s}$ runs over a codebook $\mathcal{S}$,
{\small
\begin{equation}
\label{dstbc}
\mathcal{C} = \left\{ \left[ \left[h_{1}\tilde{\textbf{s}}\textbf{A}_{1} + h_{1}^{*}\tilde{\textbf{s}}^{*}\textbf{B}_{1} \right]^{T} ~~ \left[h_{2}\tilde{\textbf{s}}\textbf{A}_{2} + h_{2}^{*}\tilde{\textbf{s}}^{*}\textbf{B}_{2}\right]^{T} ~~ \cdots ~~ \left[h_{K}\tilde{\textbf{s}}\textbf{A}_{K} + h_{K}^{*}\tilde{\textbf{s}}^{*}\textbf{B}_{K}\right]^{T} \right]^{T} \right\}
\end{equation}
}
\noindent is called the Precoded Distributed Space-Time Block code (PDSTBC) which is determined by the set $\left\lbrace \textbf{P}, \textbf{Q}, \textbf{A}_{k}, \textbf{B}_{k}\right\rbrace$.
\end{definition}
\begin{remark} From \eqref{dstbc}, every codeword of a PDSTBC includes random variables $h_{k}$ for all $k = 1, 2, \cdots K$. Even though, $h_{k}$ can take any complex value, since the destination knows the channel set $\left\lbrace h_{1}, h_{2}, \cdots h_{K}\right\rbrace$ for every codeword use, the cardinality of $\mathcal{C}$ is equal to the cardinality of $\mathcal{S}$. The properties of the PDSTBC will depend on the set $\left\lbrace \textbf{P}, \textbf{Q}, \textbf{A}_{k}, \textbf{B}_{k}\right\rbrace$ alone but not on the realisation of the channels $h_{k}$'s. In this paper, on the similar lines of \cite{ZhK}, we derive conditions on the set $\left\lbrace \textbf{P}, \textbf{Q}, \textbf{A}_{k}, \textbf{B}_{k}\right\rbrace$ such that the PDSTBC in \eqref{dstbc} is SSD for any values of $\left\lbrace h_{1}, h_{2}, \cdots h_{K}\right\rbrace$. In other words, the derived conditions are such that irrespective of the realisation of $h_{k}$'s, the PDSTBC in \eqref{dstbc} is SSD.
\end{remark}

\indent The covariance matrix $\textbf{R} \in \mathbb{C}^{T \times T}$ of the noise vector $\textbf{n}$ is given by
{\small
\begin{equation}
\label{covariance_matrix}
\textbf{R} = {\frac{P_{2}T}{(1 + P_1)N}}\left[ \sum_{k=1}^{K} |g_{k}|^{2}\left\lbrace \textbf{A}_{k}^{H}\textbf{A}_{k} + \textbf{B}_{k}^{H}\textbf{B}_{k}\right\rbrace \right]  + \textbf{I}_{T}.\\
\end{equation}
}
The Maximum Likelihood (ML) decoder decodes to a vector $\hat{\textbf{s}}$ where
{\small
\begin{equation*}
\hat{\textbf{s}} = arg\, \min_{s \in \mathcal{S}} \left[ \textbf{y} - \sqrt{\frac{P_{1}P_{2}T}{(1 + P_{1})N}}\textbf{g}\textbf{X}\right]\textbf{R}^{-1}\left[ \textbf{y} - \sqrt{\frac{P_{1}P_{2}T}{(1 + P_{1})N}}\textbf{g}\textbf{X}\right]^{H}\\
\end{equation*}
\begin{equation*}
\label{ML}
 = arg\, \min_{s \in \mathcal{S}} \left[ -2Re \left(\sqrt{\frac{P_{1}P_{2}T}{(1 + P_{1})N}}\textbf{g}\textbf{X}\textbf{R}^{-1}\textbf{y}^{H}\right) + {\frac{P_{1}P_{2}T}{(1 + P_{1})N}}\textbf{g}\textbf{X}\textbf{R}^{-1}\textbf{X}^{H}\textbf{g}^{H} \right].
\end{equation*}
}

\noindent With the above decoding metric, we give a definition for a SSD distributed space-time block code which also includes DOSTBCs studied in \cite{ZhK}.
\begin{definition}
\label{D_pdssd}
A PDSTBC, $\textbf{X}$ in variables $x_{1}, x_{2}, \cdots x_{N}$ is called a Precoded Distributed Single-Symbol Decodable STBC (PDSSDC), if it satisfies the following conditions,
\begin{itemize}
\item The entries of the $k^{th}$ row of $\textbf{X}$ are 0, $\pm~h_{k}\tilde{x}_{n}$, $\pm~ h_{k}^{*}\tilde{x}^{*}_{n}$ or multiples of these by $j$ where $j = \sqrt{-1}$ for any complex variable $h_{k}$. The complex variables $\tilde{x}_{n}$ for $1 \leq n \leq N$ are the components of the transmitted vector $\tilde{\textbf{s}}$ where $\tilde{\textbf{s}} = \left[\tilde{x}_{1}~\tilde{x}_{2}~\cdots~\tilde{x}_{N} \right]$.
\item The matrix $\textbf{X}$ satisfies the equality 
\begin{equation}
\label{PDSSD_condition}
\textbf{X}\textbf{R}^{-1}\textbf{X}^{H} = \sum_{i = 1}^{N} \textbf{W}_{i} \mbox{ with } \left[\textbf{W}_{i}\right]_{k,k} =  |h_{k}|^{2}\left(\upsilon_{i,k}^{\left(1\right)}|x_{iI}|^{2} + \upsilon_{i,k}^{\left(2\right)}|x_{iQ}|^{2}\right)
\end{equation}
where each $\textbf{W}_{i}$ is a $K \times K$ matrix with its non zero entries being functions of $x_{iI}$, $x_{iQ}$ and $h_{k}$ for all $k = 1, 2, \cdots, K$ and $\upsilon_{i,k}^{\left(1\right)}, \upsilon_{i,k}^{\left(2\right)} \in \mathbb{R}$.
\end{itemize}
\end{definition}

\indent We study the properties of the relay matrices $\textbf{A}_{k}, \textbf{B}_{k}$ and the precoding matrices $\textbf{P}$ and $\textbf{Q}$ such that the vectors transmitted simultaneously from all the relays appear as a PDSSDC at the destination. Certain properties of the relay matrices have been studied in the context of DOSTBCs in \cite{ZhK}. We recall some of the definitions and properties used in \cite{ZhK} so as to study the properties of the relay matrices of a PDSSDC. A matrix is said to be column (row) monomial, if there is atmost one non-zero entry in every column (row) of it. 
\begin{lemma}
\label{Relay_Mat_condition_1}
The relay matrices $\textbf{A}_{k}$ and $\textbf{B}_{k}$ of a PDSSDC satisfy the following conditions,
\begin{itemize}
\item The entries of $\textbf{A}_{k}$ and $\textbf{B}_{k}$ are 0, $\pm~ 1, \pm~ j$.
\item $\textbf{A}_{k}$ and $\textbf{B}_{k}$ cannot have non-zeros at the same position.
\item $\textbf{A}_{k}$, $\textbf{B}_{k}$ and $\textbf{A}_{k} + \textbf{B}_{k}$ are column monomial matrices.
\end{itemize}
\end{lemma}
\begin{proof}
The proof is on the similar lines of the proof for Lemma $1$ in \cite{ZhK}.
\end{proof}
\begin{lemma}
\label{sept_cond}
If $\textbf{A}, \textbf{C}, \textbf{D} \in \mathbb{C}^{N \times N}$ and $\textbf{s} = \left[ x_{1}, x_{2}, \cdots, x_{N}\right] \in \mathbb{C}^{1 \times N}$, with each $x_{i} =  x_{iI} + j x_{iQ}$, then
\begin{equation}
\textbf{s}\textbf{A}\textbf{s}^{H} + \textbf{s}\textbf{C}\textbf{s}^{T} + \textbf{s}^{*}\textbf{D}\textbf{s}^{H}  = \sum_{i=1}^{N} f_{i}\left(x_{iI}, x_{iQ}\right) 
\end{equation}
where $f_{i}\left(x_{iI}, x_{iQ}\right)$ is a complex valued function of the variables $x_{iI}$ and $x_{iQ}$ if and only if $\textbf{A}, ~\textbf{C} + \textbf{C}^{T}, ~\textbf{D}+ \textbf{D}^{T} \in \mathcal{D}_{N}$. 
\end{lemma}
\begin{proof}
Refer to the proof of Lemma $2$ in \cite{HaR}.
\end{proof}
Using the results of Lemma \ref{sept_cond}, in the following Theorem, we provide a set of necessary and sufficient conditions on the matrix set $\left\lbrace \textbf{P}, \textbf{Q}, \textbf{A}_{k}, \textbf{B}_{k}\right\rbrace$ such that a PDSTBC \textbf{X} with the above matrix set is a PDSSDC.
\begin{theorem}
\label{necc_suf_cond}
A PDSTBC $\textbf{X}$ is a PDSSDC if and only if the relay matrices $\textbf{A}_{k}$, $\textbf{B}_{k}$ satisfy the following conditions,\\

\noindent (i) For $1 \leq k \neq k' \leq K$,
{\small
\begin{eqnarray}
\label{cond_1}
\mathbf{\Upsilon}_{1}\textbf{A}_{k}\textbf{R}^{-1}\textbf{A}_{k'}^{H}\mathbf{\Upsilon}_{2}^{H} + \mathbf{\Pi}_{1}^{*}\textbf{A}_{k'}^{*}\textbf{R}^{-1}\textbf{A}_{k}^{T}\mathbf{\Pi}_{2}^{T} \in \mathcal{D}_{N} \mbox{ for } \left\{ \begin{array}{cccccccccc}
\mathbf{\Upsilon}_{1} = \mathbf{\Upsilon}_{2} = \textbf{P} \mbox{ and } \mathbf{\Pi}_{1} = \mathbf{\Pi}_{2} = \textbf{Q};\\
\mathbf{\Pi}_{2} = \mathbf{\Upsilon}_{1} = \textbf{P} \mbox{ and } \mathbf{\Pi}_{1} = \mathbf{\Upsilon}_{2} = \textbf{Q};\\
\mathbf{\Pi}_{1} = \mathbf{\Upsilon}_{2} = \textbf{P} \mbox{ and } \mathbf{\Pi}_{2} = \mathbf{\Upsilon}_{1} = \textbf{Q};
\end{array} 
\right.
\end{eqnarray}
\begin{eqnarray}
\label{cond_2}
\mathbf{\Upsilon}_{1}^{*}\textbf{B}_{k}\textbf{R}^{-1}\textbf{B}_{k'}^{H}\mathbf{\Upsilon}_{2}^{T} + \mathbf{\Pi}_{1}\textbf{B}_{k'}^{*}\textbf{R}^{-1}\textbf{B}_{k}^{T}\mathbf{\Pi}_{2}^{H} \in \mathcal{D}_{N}
\mbox{ for } \left\{ \begin{array}{cccccccccc}
\mathbf{\Upsilon}_{1} = \mathbf{\Upsilon}_{2} = \textbf{Q} \mbox{ and } \mathbf{\Pi}_{1} = \mathbf{\Pi}_{2} = \textbf{P};\\
\mathbf{\Pi}_{2} = \mathbf{\Upsilon}_{1} = \textbf{Q} \mbox{ and } \mathbf{\Pi}_{1} = \mathbf{\Upsilon}_{2} = \textbf{P};\\
\mathbf{\Pi}_{1} = \mathbf{\Upsilon}_{2} = \textbf{Q} \mbox{ and } \mathbf{\Pi}_{2} = \mathbf{\Upsilon}_{1} = \textbf{P}.
\end{array}
\right.
\end{eqnarray}
}
(ii) For $1 \leq k, k' \leq K$,
{\small
\begin{equation}
\label{cond_3}
\mathbf{\Pi}^{*}\left[ \textbf{B}_{k}\textbf{R}^{-1}\textbf{A}_{k'}^{H} + \textbf{A}_{k'}^{*}\textbf{R}^{-1}\textbf{B}_{k}^{T}\right]\mathbf{\Upsilon}^{H} \in \mathcal{D}_{N},
\mbox{for } \left\{ \begin{array}{cccccccccc}
\mathbf{\Upsilon} =  \textbf{P} \mbox{ and } \mathbf{\Pi} = \textbf{Q};\\
\mathbf{\Upsilon} =  \textbf{P} \mbox{ and } \mathbf{\Pi} = \textbf{P};\\
\mathbf{\Upsilon} =  \textbf{Q} \mbox{ and } \mathbf{\Pi} = \textbf{Q};
\end{array}
\right.
\end{equation}
\begin{equation}
\label{cond_4}
\mathbf{\Pi} \left[\textbf{A}_{k}\textbf{R}^{-1}\textbf{B}_{k'}^{H} + \textbf{B}_{k'}^{*}\textbf{R}^{-1}\textbf{A}_{k}^{T}\right]\mathbf{\Upsilon}^{T} \in \mathcal{D}_{N}, \mbox{for } \left\{ \begin{array}{cccccccccc}
\mathbf{\Upsilon} =  \textbf{Q} \mbox{ and }  \mathbf{\Pi} = \textbf{P};\\
\mathbf{\Upsilon} =  \textbf{P} \mbox{ and } \mathbf{\Pi} = \textbf{P};\\
\mathbf{\Upsilon} =  \textbf{Q} \mbox{ and } \mathbf{\Pi} = \textbf{Q}.
\end{array}
\right. 
\end{equation}
}
(iii) For $1 \leq k \leq K$,
{\small
\begin{equation}
\label{cond_5}
\textbf{A}_{k}\textbf{R}^{-1}\textbf{A}_{k}^{H} + \textbf{B}_{k}^{*}\textbf{R}^{-1}\textbf{B}_{k}^{T} = \mbox{diag}\left[D_{1,k}, D_{2,k}, \cdots, D_{N,k} \right].
\end{equation}
}
\indent where $D_{n,k} \in \mathbb{R}$ for all $n = 1, 2, \cdots N$.
\end{theorem}
\begin{proof}
Refer to the proof of Lemma $3$ in \cite{HaR}.
\end{proof}

\indent Theorem \ref{necc_suf_cond} provides a set of necessary and sufficient conditions on the relay matrices $\textbf{A}_{k}, \textbf{B}_{k}$ and the precoding matrices $\textbf{P}$ and $\textbf{Q}$ such that, $\textbf{X}$ is a PDSSDC. The matrices, $\mathbf{\Upsilon}_{1}\textbf{A}_{k}\textbf{R}^{-1}\textbf{A}_{k'}^{H}\mathbf{\Upsilon}_{2}^{H} + \mathbf{\Pi}_{1}^{*}\textbf{A}_{k'}^{*}\textbf{R}^{-1}\textbf{A}_{k}^{T}\mathbf{\Pi}_{2}^{T},~$ $\mathbf{\Upsilon}_{1}^{*}\textbf{B}_{k}\textbf{R}^{-1}\textbf{B}_{k'}^{H}\mathbf{\Upsilon}_{2}^{T} + \mathbf{\Pi}_{1}\textbf{B}_{k'}^{*}\textbf{R}^{-1}\textbf{B}_{k}^{T}\mathbf{\Pi}_{2}^{H},~$  $\mathbf{\Pi}^{*}\left[ \textbf{B}_{k}\textbf{R}^{-1}\textbf{A}_{k'}^{H} + \textbf{A}_{k'}^{*}\textbf{R}^{-1}\textbf{B}_{k}^{T}\right]\mathbf{\Upsilon}^{H} \mbox{ and } \mathbf{\Pi} \left[\textbf{A}_{k}\textbf{R}^{-1}\textbf{B}_{k'}^{H} + \textbf{B}_{k'}^{*}\textbf{R}^{-1}\textbf{A}_{k}^{T}\right]\mathbf{\Upsilon}^{T}$ in the conditions of \eqref{cond_1} - \eqref{cond_4} need to be diagonal. This implies that the above matrices can also be $\textbf{0}_{N}$. The DOSTBCs studied in \cite{ZhK} are a special class of PDSSDCs since the relay matrices of DOSTBCs $\left( \mbox{Lemma}~1, \cite{ZhK} \right)$ satisfy the conditions of Theorem \ref{necc_suf_cond}. In particular, the necessary and sufficient conditions on the relay matrices of DOSTBCs as shown in $\mbox{Lemma}~1$ of $\cite{ZhK}$ can be obtained from the necessary and sufficient conditions of PDSSDCs by making $\textbf{P} = \textbf{I}_{N}$, $\textbf{Q} = \textbf{0}_{N}$ and $\mathcal{D}_{N} = \textbf{0}_{N}$ in \eqref{cond_1} - \eqref{cond_4}.\\
\indent A PDSSDC, $\textbf{X}$ in variables $x_{1}, x_{2}, \cdots x_{N}$ can be written in the form of a linear dispersion code \cite{HaH} as $\textbf{X} = \sum_{j = 1}^{N} x_{iI} \Phi_{iI} + x_{iQ} \Phi_{iQ}$ where $\Phi_{iI}, \Phi_{iQ} \in \mathbb{C}^{K \times T}$ are called the weight matrices of $\textbf{X}$. Within the class of PDSSDCs, we consider a special set of codes called Unitary PDSSDCs defined as,
\begin{definition}
\label{u_p_dssd}
A PDSSDC, $\textbf{X}$ is called a Unitary PDSSDC, if the weight matrices of $\textbf{X}$ satisfies the following conditions, $\Phi_{iI}\Phi_{iI}^{H},~\Phi_{iQ}\Phi_{iQ}^{H} \in \mathcal{D}_{K}^{\mathbf{+}}$ for all $i = 1, 2, \cdots N.$
\end{definition}
\begin{remark}
We caution the reader to note the difference between the definition for a Unitary PDSSDC for cooperative networks and  the definition for a Unitary SSD code for MIMO systems \cite{SaR}. For better clarity, we recall the definition for a Unitary SSD code designed for MIMO systems.
A SSD STBC, $\tilde{\textbf{X}}$ in variables $x_{1}, x_{2}, \cdots x_{N}$ when written in the form of a linear dispersion code is given by $\tilde{\textbf{X}} = \sum_{j = 1}^{N} x_{iI} \tilde{\Phi}_{iI} + x_{iQ} \tilde{\Phi}_{iQ}$,
where $\tilde{\Phi}_{iI}, \tilde{\Phi}_{iQ} \in \mathbb{C}^{K \times T}$ are called the weight matrices of $\tilde{\textbf{X}}$. The design $\tilde{\textbf{X}}$ is said to be a unitary SSD if $\tilde{\Phi}_{iI}\tilde{\Phi}_{iI}^{H} = \textbf{I}_{K}$ for all $i = 1, 2, \cdots N$. The difference between the two definitions can be observed as the definition for a unitary SSD STBC is a special case of the definition for a unitary PDSSDC.
\end{remark}

\indent It can be verified that DOSTBCs belong to the class of Unitary PDSSDCs. In the rest of the paper, we consider only unitary PDSSDCs. However, it is to be noted that the class of non-unitary PDSSDCs is not empty. A class of low decoding complexity DSTBCs called Precoded Coordinate Interleaved Orthogonal Design (PCIOD) has been introduced in \cite{RaR1} wherein the authors have proposed a design, $\textbf{X}_{PCIOD}$ for a network with 4 relays which is SSD (Example 1 of \cite{RaR1}). It can be observed that the proposed code $\textbf{X}_{PCIOD}$ given in \eqref{pciod} belongs to the class of non-unitary PDSSDCs. Since we consider only unitary PDSSDCs, through out the paper a PDSSDC is meant unitary PDSSDC.
{\small
\begin{equation}
\label{pciod}
\textbf{X}_{PCIOD} = \left[\begin{array}{rrrr}
h_{1}\tilde{x}_{1} & h_{1}\tilde{x}_{2} & 0 & 0\\
- h_{2}^{*}\tilde{x}_{2}^{*} & h_{2}^{*}\tilde{x}_{1}^{*} & 0 & 0\\
0 & 0 & h_{3}\tilde{x}_{3} & h_{3}\tilde{x}_{4}\\
0 & 0 & - h_{4}^{*}\tilde{x}_{4}^{*} & h_{4}^{*}\tilde{x}_{3}^{*}\\
\end{array}\right].
\end{equation}
}
The precoding matrices, $\textbf{P}$ and $\textbf{Q}$ required at the source to construct $\textbf{X}_{PCIOD}$ are 
{\small
\begin{equation*}
\textbf{P} = \frac{1}{2}\left[\begin{array}{rrrr}
1 & 0 & 1 & 0\\
0 & 1 & 0 & 1\\
1 & 0 & 1 & 0\\
0 & 1 & 0 & 1\\
\end{array}\right] ;~ \textbf{Q} = \frac{1}{2}\left[\begin{array}{rrrr}
1 & 0 & -1& 0\\
0 & 1 & 0 & -1\\
-1 & 0 & 1& 0\\
0 & -1 & 0 & 1\\
\end{array}\right].
\end{equation*}
}
\indent Various class of single-symbol decodable STBCs for cooperative networks are captured in Figure \ref{various_DSSD} which is first partitioned in to two sets depending on whether the codes are unitary or non-unitary (Definition \ref{u_p_dssd}). The class of PDSSDCs are shown to be a subset of the class of SSD codes for cooperative networks. The set of unitary distributed SSD codes are shown to contain the DOSTBCs and the S-PDSSDCs (Definition \ref{D_so_updssd}). An example of a code which belongs to the class of non-unitary Distributed SSD codes but not to the class of PDSSDCs is given below,
\begin{equation}
\textbf{X}_{DSSDC} = \left[\begin{array}{cc}
\mbox{Re} (h_{1}x_{1}) + j \mbox{Im} (h_{1}x_{2}) & - \mbox{Re} (h_{1}x_{2}) + j \mbox{Im} (h_{1}x_{1})\\
\mbox{Re} (h_{2}x_{2}) + j \mbox{Im} (h_{2}x_{1}) & \mbox{Re} (h_{2}x_{1}) - j \mbox{Im} (h_{2}x_{2})\\
\end{array}\right].
\end{equation}
{\small
\begin{equation*}
\textbf{A}_{1} = \frac{1}{2}\left[\begin{array}{rrrr}
1 & 1\\
1  & -1\\
\end{array}\right] ;~ \textbf{B}_{1} = \frac{1}{2}\left[\begin{array}{rrrr}
1 & -1\\
-1  & -1\\
\end{array}\right]; ~ \textbf{A}_{2} = \frac{1}{2}\left[\begin{array}{rrrr}
1 & 1\\
1  & -1\\
\end{array}\right]~\mbox{ and }~\textbf{B}_{2} = \frac{1}{2}\left[\begin{array}{rrrr}
-1 & 1\\
 1 & 1\\
\end{array}\right].
\end{equation*}
}

\noindent From the above matrices, it can be verified that, $\textbf{R}^{-1}$ is a scaled identity matrix. Therefore, $\textbf{X}\textbf{R}^{-1}\textbf{X}^{H} = \textbf{R}^{-1}\textbf{X}\textbf{X}^{H}$ and $\textbf{X}\textbf{X}^{H}$ is given by,
\begin{equation*}
\textbf{X}\textbf{X}^{H} = \left[\begin{array}{cccc}
|h_{1}|^{2}\sum_{i = 1}^{2} |x_{i}|^{2} & \left(h_{1}^{*}h_{2} - h_{2}^{*}h_{1}\right) \sum_{i = 1}^{2} |x_{i}|^{2}\\
\left(h_{2}^{*}h_{1} - h_{1}^{*}h_{2}\right) \sum_{i = 1}^{2} |x_{i}|^{2}  & |h_{2}|^{2}\sum_{i = 1}^{2} |x_{i}|^{2}\\
\end{array}\right].
\end{equation*}
\section{Semi-orthogonal PDSSDC}\label{sec3}
\indent From the definition of a PDSSDC (Definition \ref{D_pdssd}), $\left[\textbf{X}\textbf{R}^{-1}\textbf{X}^{H}\right]_{k, k'}$ for any $k \neq k'$ can be non-zero. i.e, the $k^{th}$ and the $k'^{th}$ row of a PDSSDC $\textbf{X}$, need not satisfy the equality $\left[\textbf{X}\textbf{R}^{-1}\textbf{X}^{H}\right]_{k, k'} = 0$, but $\left[\textbf{X}\textbf{R}^{-1}\textbf{X}^{H}\right]_{k, k'}$ must be a complex linear combination of several terms with each term being a function of in-phase and quadrature component of a single information variable. Through out the paper, the $k^{th}$ and the $k'^{th}$ row of a PDSSDC are referred to as \textbf{R}-orthogonal if $\left[\textbf{X}\textbf{R}^{-1}\textbf{X}^{H}\right]_{k, k'} = 0$. Similarly, the $k^{th}$ and the $k'^{th}$ row are referred to as \textbf{R}-non-orthogonal if $\left[\textbf{X}\textbf{R}^{-1}\textbf{X}^{H}\right]_{k, k'} \neq 0$. In this paper, we identify a special class of PDSSDCs where every row of $\textbf{X}$ is \textbf{R}-non-orthogonal to atmost one of its rows and we formally define it as,
\begin{definition}\label{D_so_updssd}
A PDSSDC is said to be a Semi-orthogonal PDSSDC (S-PDSSDC) if every row of a PDSSDC is \textbf{R}-non-orthogonal to atmost one of its rows.
\end{definition}
\indent From the above definition, it can be observed that DOSTBCs are a proper subclass of S-PDSSDCs since every row of DOSTBC is \textbf{R}-orthogonal to every other row. The definition of a S-PDSSDC implies that the set of $K$ rows can be partitioned in to atleast $\lceil{\frac{K}{2}}\rceil$ groups such that every group has atmost two rows.\\
\indent The co-variance matrix, $\textbf{R}$ in \eqref{covariance_matrix} is a function of (i) the realisation of the channels from the relays to the destination and (ii) the relay matrices, $\textbf{A}_{k}, \textbf{B}_{k}$. In general, $\textbf{R}$ may not be diagonal in which case the construction of S-PDSSDCs is not straight forward. On the similar lines of \cite{ZhK}, we consider a subset of S-PDSSDCs whose covariance matrix is diagonal and refer to such a subset as row monomial S-PDSSDCs (RS-PDSSDCs). It can be proved that the relay matrices of a RS-PDSSDC are row monomial if and only if the corresponding covariance matrix is diagonal (refer to Theorem 1 of \cite{ZhK}). The row monomial property of the relay matrices implies that every row of a RS-PDSSDC contains the variables $\pm~ h_{k}\tilde{x}_{n}$ and $\pm~ h_{k}^{*}\tilde{x}^{*}_{n}$ atmost once for all $n$ such that $1 \leq n \leq N$.
\subsection{upperbound on the symbol-rate of RS-PDSSDCs}

\indent In this subsection, we derive an upperbound on the rate of RS-PDSSDCs in symbols per channel use in the second phase i.e an upperbound on $\frac{N}{T}$. Towards that end, properties of the relay matrices $\textbf{A}_{k}$, $\textbf{A}_{k'}$, $\textbf{B}_{k}$ and $\textbf{B}_{k'}$ of RS-PDSSDC are studied when the rows corresponding to the indices $k$ and $k'$ are (i) \textbf{R}-orthogonal and (ii) \textbf{R}-non-orthogonal. For the former case, the properties of $\textbf{A}_{k}$, $\textbf{A}_{k'}$, $\textbf{B}_{k}$ and $\textbf{B}_{k'}$ have been studied in \cite{ZhK}. If $k$ and $k'$ represent the indices of the rows of a RS-PDSSDC that are \textbf{R}-orthogonal, then the corresponding relay matrices $\textbf{A}_{k}, \textbf{A}_{k'}, \textbf{B}_{k}$ and $\textbf{B}_{k'}$ satisfies the following conditions (i) $\textbf{A}_{k}$ and $\textbf{A}_{k'}$ are column disjoint and (ii) $\textbf{B}_{k}$ and $\textbf{B}_{k'}$ are column disjoint. (Lemma 3 of \cite{ZhK}) i.e., the matrices $\textbf{A}_{k}$ and $\textbf{A}_{k'}$ cannot contain non-zero entries on the same columns simultaneously. The above result implies,
\begin{equation}
\label{orth_cond_As}
\textbf{A}_{k}\textbf{A}_{k'}^{H} = \textbf{0}_{N} ~~\mbox{ and }~~ \textbf{B}_{k}^{*}\textbf{B}_{k'}^{T} = \textbf{0}_{N}.
\end{equation}
In order to address the latter case, consider the 2 $\times$ 2 matrix $\mathbf{\Xi}$ as given below, 
\begin{equation*}
\mathbf{\Xi} = \left[\begin{array}{cc}
h_{k}\tilde{x}_{i}  &  h_{k'}\tilde{x}_{m}\\
h_{k}\blacklozenge  & h_{k'}\clubsuit \\
\end{array}\right]
\end{equation*}
where $h_{k}, h_{k'}$ are complex random variables. The complex variables $\tilde{x}_{i}$ and $\tilde{x}_{m}$ are the components of the transmitted vector $\tilde{s}$ (as in \eqref{cord_int}) where $\tilde{\textbf{s}} = \left[\tilde{x}_{1}~\tilde{x}_{2}~\cdots~\tilde{x}_{N} \right].$ In particular, the complex variables $\tilde{x}_{i} \mbox{ and }\tilde{x}_{m}$ are of the form, $\tilde{x}_{i} = \pm~ x_{\gamma \square}\pm~ j x_{\lambda \square} ~ \mbox{ and }~\tilde{x}_{m} = \pm~ x_{\delta \square} \pm~ j x_{\mu \square}$ where 
\begin{itemize}
\item $\gamma, \lambda, \delta$ and $\mu$ are positive integers such that $ 1 \leq \gamma, \lambda, \delta$, $\mu \leq N$ and atmost any two of these integers can be equal. 
\item The subscript $\square$ denotes either $I$ (in-phase component) or $Q$ (quadrature component) of a variable and \item $\blacklozenge$, $\clubsuit$ are indeterminate complex variables which can take values of the form $\pm~ \tilde{x}_{n}$ or $\pm~\tilde{x}_{n}^{*}$ such that $1 \leq n \leq N$. 
\end{itemize}
For example, if $N$ = 4, $\tilde{x}_{i}$ and $\tilde{x}_{m}$ can possibly be $ x_{2 I} + j x_{3 Q}$ and $ x_{3 I} + j x_{4 Q}$ respectively.\\
\indent In Lemma \ref{SSD_Cols}, we investigate various choices on the indeterminate variables $\blacklozenge$ and $\clubsuit$ such that $\left[\mathbf{\Xi}^{H}\mathbf{\Xi}\right]_{1,2}$ is a complex linear combination of several terms with each term being a function of in-phase and quadrature components of a single information variable. In general, the real variables $x_{\gamma \square}, x_{\lambda \square}, x_{\delta \square}$ and $x_{\mu \square}$ can appear in $ \tilde{x}_{i} $ and $ \tilde{x}_{m} $ with arbitrary signs. With out loss of generality, we assume that $ \tilde{x}_{i} $ and $ \tilde{x}_{m} $ are given by
\begin{eqnarray}
\label{variables}
\tilde{x}_{i} = x_{\gamma \square} + j x_{\lambda \square} ~ \mbox{ and } ~~ \tilde{x}_{m} = x_{\delta \square} + j x_{\mu \square}.
\end{eqnarray}
However, the results of Lemma \ref{SSD_Cols} will continue to hold even if the variables $x_{\gamma \square}, x_{\lambda \square}, x_{\delta \square}$ and $x_{\mu \square}$ appear in $ \tilde{x}_{i} $ and $ \tilde{x}_{m} $ with any arbitrary signs. Since a RS-PDSSDC takes variables only of the form $\pm~ h_{k}\tilde{x}_{n}$, $\pm~ h_{k}^{*}\tilde{x}^{*}_{n}$ and every row of a RS-PDSSDC contains the variables $\pm~ h_{k}\tilde{x}_{n}$ and $\pm~ h_{k}^{*}\tilde{x}^{*}_{n}$ atmost once, we have the following restrictions on the choice of the indeterminate variables $\blacklozenge$ and $\clubsuit$ that (i) the indeterminate $\blacklozenge$ cannot take the variable $ \tilde{x}_{i} $ and variables of the form $ \tilde{x}_{n}^{*} $ for all $n = 1, 2, \cdots N$ and (ii) the indeterminate $\clubsuit$ cannot take the variable $ \tilde{x}_{m} $ and variables of the form $ \tilde{x}_{n}^{*} $ for all $n = 1, 2, \cdots N$. 
\begin{lemma}
\label{SSD_Cols}
If there exists a solution on the choice of $\blacklozenge$ and $\clubsuit$ such that
{\small
\begin{eqnarray}
\label{SSD_Cols_cond}
\left[\mathbf{\Xi}^{H}\mathbf{\Xi}\right]_{1,2} = f_{1}\left(x_{\delta I}, x_{\delta Q}, h_{k}, h_{k'}\right) + f_{2}\left(x_{\gamma I}, x_{\gamma Q}, h_{k}, h_{k'}\right) \nonumber\\
~~~~~~~~~~~~~~~~~ + f_{3}\left(x_{\lambda I}, x_{\lambda Q}, h_{k}, h_{k'}\right) + f_{4}\left(x_{\mu I}, x_{\mu Q}, h_{k}, h_{k'}\right) \neq 0, 
\end{eqnarray}
}
\noindent then only one of the following is true,\\
(i) $\delta = \gamma \mbox{ and }  \mu =  \lambda$.\\ 
(ii) $\delta =\lambda \mbox{ and } \mu =  \gamma$.\\
where $f_{i}\left(x_{\beta I}, x_{\beta Q}, h_{k}, h_{k'}\right)$ is a complex valued function of the variables, $x_{\beta I}, x_{\beta Q}, h_{k}, h_{k'}$ for all $i = 1,2, \cdots 4$ and $\beta = \mu, \lambda, \gamma, \delta.$
\end{lemma}
\begin{proof}
Refer to the proof of Lemma $4$ in \cite{HaR}.
\end{proof}
Similarly, it can be shown that, the results of Lemma \ref{SSD_Cols} holds true even if the matrix $\Xi$ is of the form,
\begin{equation*}
\left[\begin{array}{cc}
h_{k}^{*}\tilde{x}_{i}^{*}  &  h_{k'}^{*}\tilde{x}_{m}^{*}\\
h_{k}^{*}\blacklozenge  & h_{k'}^{*}\clubsuit \\
\end{array}\right].
\end{equation*}
We use the results of Lemma \ref{SSD_Cols} to study the properties of the relay matrices of a RS-PDSSDC.
\begin{lemma}
\label{SSD_Cols_PDSSD}
Let $\textbf{A}_{k}$ and $\textbf{A}_{k'}$ be the relay matrices of a RS-PDSSDC, $\textbf{X}$. If $\left[ \textbf{A}_{k}\textbf{A}_{k'}^{H}\right]_{i,m}$ is a non zero entry for $i \neq m$, then the precoding matrices at the source $\textbf{P}$ and $\textbf{Q}$ are such that
\begin{equation}
\label{SSD_Cols_PDSSD_cond}
\tilde{x}_{iI}, \tilde{x}_{iQ}, \tilde{x}_{mI} \mbox{ and }\tilde{x}_{mQ} \in \left\lbrace {x}_{nI}, {x}_{nQ}, {x}_{n'I}, {x}_{n'Q} \right\rbrace \mbox{ with }
\end{equation}
\begin{equation*}
\tilde{x}_{iI}, \tilde{x}_{iQ} \in \left\lbrace {x}_{n \square}, {x}_{n' \square}\right\rbrace ~ \mbox{ and }~\tilde{x}_{mI}, \tilde{x}_{mQ} \in \left\lbrace {x}_{n \square}, {x}_{n' \square}\right\rbrace
\end{equation*}
for some $n \neq n'$ where $1 \leq n,n' \leq N$ and the subscript $\square$ represents either $I$ or $Q$.
\end{lemma}
\begin{proof}
Refer to the proof of Lemma $5$ in \cite{HaR}.
\end{proof}
\begin{lemma}
Let $\textbf{B}_{k}$ and $\textbf{B}_{k'}$ be the relay matrices of a RS-PDSSDC, $\textbf{\textbf{X}}$. If $\left[ \textbf{B}_{k}^{*}\textbf{B}_{k'}^{T}\right]_{i,m}$ is a non zero entry for $i \neq m$, precoding matrices at the source $\textbf{P}$ and $\textbf{Q}$ are such that
\begin{equation}
\label{SSD_Cols_PDSSD_cond_2}
\tilde{x}_{iI}, \tilde{x}_{iQ}, \tilde{x}_{mI} \mbox{ and }\tilde{x}_{mQ} \in \left\lbrace {x}_{nI}, {x}_{nQ}, {x}_{n'I}, {x}_{n'Q} \right\rbrace \mbox{ with }
\end{equation}
\begin{equation*}
\tilde{x}_{iI}, \tilde{x}_{iQ} \in \left\lbrace {x}_{n \square}, {x}_{n' \square}\right\rbrace ~ \mbox{ and } ~\tilde{x}_{mI}, \tilde{x}_{mQ} \in \left\lbrace {x}_{n \square}, {x}_{n' \square}\right\rbrace
\end{equation*}
for some $n \neq n'$ where $1 \leq n,n' \leq N$ and the subscript $\square$ represents either $I$ or $Q$.
\end{lemma}
\begin{proof} The result can be proved on the similar lines of the proof for Lemma \ref{SSD_Cols_PDSSD}.
\end{proof}
\begin{corollary}
\label{even_non_zeros}
For a RS-PDSSDC, if $\left[ \textbf{A}_{k}\textbf{A}_{k'}^{H}\right]_{i,m}$ is non-zero, then so is $\left[ \textbf{A}_{k}\textbf{A}_{k'}^{H}\right]_{m,i}$. 
\end{corollary}
\begin{proof}
Follows from the proof for Lemma \ref{SSD_Cols} and Lemma \ref{SSD_Cols_PDSSD}.
\end{proof}
\indent From the definition of a PDSSDC (Definition \ref{D_pdssd}), non-zero entries of the $k^{th}$ row contains variables of the form $\pm~ h_{k}\tilde{x}_{n}$, $\pm~ h_{k}^{*}\tilde{x}^{*}_{n}$ or multiples of these by $j$. Therefore,
 \begin{eqnarray*}
\left[ \textbf{X}\textbf{X}^{H}\right]_{k, k} = |h_{k}|^{2}\left[\tilde{\textbf{s}}\textbf{A}_{k}\textbf{A}_{k}^{H}\tilde{\textbf{s}}^{H} + \tilde{\textbf{s}}^{*}\textbf{B}_{k}\textbf{B}_{k}^{H}\tilde{\textbf{s}}^{T}\right] + h_{k}h_{k}\left[\tilde{\textbf{s}}\textbf{A}_{k}\textbf{B}_{k}^{H}\tilde{\textbf{s}}^{T}\right] \\ + ~h_{k}^{*}h_{k}^{*}\left[\tilde{\textbf{s}}^{*}\textbf{B}_{k}\textbf{A}_{k}^{H}\tilde{\textbf{s}}^{H}\right]  = \sum_{i = 1}^{N} |h_{k}|^{2}\left(\omega_{i,k}^{(1)}|x_{iI}|^{2} + \omega_{i,k}^{(2)}|x_{iQ}|^{2}\right)
\end{eqnarray*}
where $\omega_{i,k}^{(1)}, \omega_{i,k}^{(2)} \in \mathbb{R}^{+}$ for all $k = 1, 2, \cdots, K$. From the results of Lemma $1$ in \cite{LiX}, we have
\begin{equation}
\label{row_norm}
\textbf{A}_{k}\textbf{A}_{k}^{H} + \textbf{B}_{k}^{*}\textbf{B}_{k}^{T} = \mbox{diag}\left[E_{1,k}, E_{2,k}, \cdots E_{n,k} \right]
\end{equation}
where $E_{n, k}$ are strictly positive real numbers.
\begin{lemma}
\label{main_thm}
Let $k$ and $k'$ represent the indices of the rows of a RS-PDSSDC, that are \textbf{R}-non-orthogonal, then the corresponding relay matrices $\textbf{A}_{k}, \textbf{B}_{k}, \textbf{A}_{k'}$ and $\textbf{B}_{k'}$ satisfy the following conditions,
\begin{itemize}
\item $\left[ \textbf{A}_{k}\textbf{A}_{k'}^{H}\right]_{i,i}  = \left[\textbf{B}_{k}^{*}\textbf{B}_{k'}^{T}\right]_{i,i}  = 0$ for all $i = 1, 2 \cdots N.$
\item $\textbf{A}_{k}\textbf{A}_{k'}^{H}$ and $\textbf{B}_{k}^{*}\textbf{B}_{k'}^{T}$ are both column and row monomial matrices.
\item $\textbf{A}_{k}\textbf{A}_{k'}^{H} + \textbf{B}_{k}^{*}\textbf{B}_{k'}^{T}$ is column and row monomial matrix.
\item The number of non-zero entries in $\textbf{A}_{k}\textbf{A}_{k'}^{H} + \textbf{B}_{k}^{*}\textbf{B}_{k'}^{T}$ is even.
\item The matrices $\tilde{\textbf{A}}_{k,k'}$ and $\tilde{\textbf{B}}_{k,k'}$ given by $\tilde{\textbf{A}}_{k,k'} = \left[\textbf{\textbf{A}}_{k}^{T} ~ \textbf{\textbf{A}}_{k'}^{T} \right]^{T}$ and $\tilde{\textbf{B}}_{k,k'} = \left[\textbf{\textbf{B}}_{k}^{T} ~ \textbf{\textbf{B}}_{k'}^{T} \right]^{T}$ satisfy the following inequality :
\begin{equation}
\mbox{ Rank}\left[\tilde{\textbf{A}}_{k,k'}\tilde{\textbf{A}}_{k,k'}^{H} + \tilde{\textbf{B}}_{k,k'}^{*}\tilde{\textbf{B}}_{k,k'}^{T}\right] \geq ~ \left\{ \begin{array}{cccccc}
2m &\mbox{if}& N = 2m \mbox{ and}\\
2m + 2 &\mbox{ if }& N = 2m + 1
\end{array}
\right.
\end{equation}
where $m$ is a positive integer.
\end{itemize}
\end{lemma}
\begin{proof}
Refer to the proof of Lemma $7$ in \cite{HaR}.
\end{proof}
\indent Using the properties of relay matrices $\textbf{A}_{k}$, $\textbf{A}_{k'}$, $\textbf{B}_{k}$ and $\textbf{B}_{k'}$ of a RS-PDSSDC corresponding to two different rows that are (i) \textbf{R}-orthogonal and (ii) \textbf{R}-non-orthogonal, an upperbound on the maximum rate, $\frac{N}{T}$ is derived in the following theorem.
\begin{theorem}
\label{rate_thm}
The symbol-rate of a RS-PDSSDC satisfies the inequality :
{\small
\begin{equation}
\label{Rate_UpperBound}
\mbox{ Rate } = \frac{N}{T} ~ \leq  ~ \left\{ \begin{array}{cccccc}
\frac{2}{l}  & \mbox{if}& N = 2m, K = 2l\\
\frac{2}{l + 1}  & \mbox{if}& N = 2m, K = 2l + 1\\
\frac{2m + 1}{(m + 1)l}  & \mbox{if}& N = 2m + 1, K = 2l\\
\frac{4m + 2}{(2m + 2)l + 2m + 1}  & \mbox{if}& N = 2m + 1, K = 2l + 1.
\end{array}
\right.
\end{equation}
}
where $l$ and $m$ are positive integers.
\end{theorem}
\begin{proof}
Refer to the proof of Theorem 1 in \cite{HaR}.
\end{proof}
\section{Construction of RS-PDSSDCs}\label{sec4}
In this section, we construct RS-PDSSDCs when the number of relays $K \geq 4$. The construction provides codes achieving the upperbound in \eqref{Rate_UpperBound} when (i) $N$ and $K$ are multiples of 4 and (ii) $N$ is a multiple of 4 and $K$ is 3 modulo 4. For the rest of the values of $N$ and $K$, codes meeting the upperbound are not known. In particular, for values of $N < 4$ and any $K$, the authors are not aware of RS-PDSSDCs with rates higher than that of row monomial DOSTBCs. The following construction provides RS-PDSSDCs with rates higher than that of row monomial DOSTBCs when $N \geq 4$ and $K \geq 4$. We first provide the construction of the precoding matrices \textbf{P} and \textbf{Q} and then present the construction of RS-PDSSDCs.
\subsection{Construction of precoding matrices $\textbf{P}$ and $\textbf{Q}$}\label{cont_p_q}
Let $\mathbf{\Gamma}, \mathbf{\Omega} \in \mathbb{C}^{4 \times 4}$ be given by
{\small
\begin{equation*}
\mathbf{\Gamma} = \frac{1}{2}\left[\begin{array}{rrrr}
1 & 0 & -j & 0\\
0  & 1 & 0 & -j\\
0 & 1 & 0 & j\\
1 & 0 & j & 0\\
\end{array}\right] \mbox{ and } \mathbf{\Omega} = \frac{1}{2} \left[\begin{array}{rrrr}
1 & 0 & j & 0\\
0  & 1 & 0 & j\\
0 & -1 & 0 & j\\
-1  & 0 & j & 0\\
\end{array}\right].
\end{equation*}
}
Let $N = 4y + a$, where $a$ can take values of  $0, 1, 2 \mbox{ and } 3$ and $y$ is any positive integer. For a given value of $a$ and $y$, the precoding matrices $\textbf{P}$ and $\textbf{Q}$ at the source are constructed as,
{\small
\begin{equation*}
\textbf{P} = \left[\begin{array}{rrrr}
\mathbf{\Gamma} \otimes \textbf{I}_{y} & \textbf{0}_{4y \times a}\\
\textbf{0}_{a \times 4y}  & \textbf{I}_{a}\\
\end{array}\right];~ \textbf{Q} = \left[\begin{array}{rrrr}
\mathbf{\Omega} \otimes \textbf{I}_{y} & \textbf{0}_{4y \times a}\\
\textbf{0}_{a \times 4y}  & \textbf{0}_{a}\\
\end{array}\right].
\end{equation*}
}
\begin{example} 
For $N = 6,$ we have $y = 1 \mbox{ and } a =2$. Following the above construction method, precoding matrices $\textbf{P}$ and $\textbf{Q}$ are given by,
{\small
\begin{equation*}
\textbf{P} = \frac{1}{2}\left[\begin{array}{rrrrrr}
1 & 0 & -j & 0 & 0 & 0\\
0  & 1 & 0 & -j & 0 & 0\\
0 & 1 & 0 & j & 0 & 0\\
1  & 0 & j & 0 & 0 & 0\\
0 & 0 & 0 & 0 & 2 & 0\\
0 & 0 & 0 & 0 & 0 & 2\\
\end{array}\right];~ \textbf{Q} = \frac{1}{2} \left[\begin{array}{rrrrrr}
1 & 0 & j & 0 & 0 & 0\\
0  & 1 & 0 & j & 0 & 0\\
0 & -1 & 0 & j & 0 & 0\\
-1 & 0 & j & 0 & 0 & 0\\
0 & 0 & 0  & 0 & 0 & 0\\
0 & 0 & 0 & 0 & 0 & 0\\
\end{array}\right].
\end{equation*}
}
\end{example}
\subsection{Construction of RS-PDSSDCs}
Through out this subsection, we denote a RS-PDSSDC for $K$ relays with $N$ variables as $\textbf{X}\left(N, K \right)$. Construction of RS-PDSSDCs is divided in to three cases depending on the values of $N$ and $K$.

\begin{case}\label{case1} $N = 4y$ and $K = 4x$ : In this case, we construct RS-PDSSDCs in the following 4 steps.\\
Step (i) : Let $\textbf{U}_{x_{1},~ x_{2}}$ be a $2 \times 2$ Alamouti design in complex variables $x_{1}$, $x_{2}$ as given below,
\begin{equation}
\label{alamouti}
\textbf{U}_{x_{1},~ x_{2}}  = \left[\begin{array}{rr}
x_{1} & x_{2}\\
- x_{2}^{*}  & x_{1}^{*}\\
\end{array}\right].
\end{equation}
Using the design in \eqref{alamouti}, construct a $4 \times 4$ design, $\mathbf{\Omega}_{m}$ in 4 complex variables  $\tilde{x}_{4m + 1}, \tilde{x}_{4m + 2}, \tilde{x}_{4m + 3}$ and $\tilde{x}_{4m + 4}$ as shown below for all $m =  0, 1, \cdots y - 1$.
{\small
\begin{equation*}
\mathbf{\Omega}_{m} = \left[\begin{array}{rr}
\textbf{U}_{\tilde{x}_{4m + 1},~ \tilde{x}_{4m + 2}} & \textbf{U}_{\tilde{x}_{4m + 3},~ \tilde{x}_{4m + 4}}\\
\textbf{U}_{\tilde{x}_{4m + 3},~ \tilde{x}_{4m + 4}}  & \textbf{U}_{\tilde{x}_{4m + 1},~ \tilde{x}_{4m + 2}}\\
\end{array}\right] = \left[\begin{array}{rrrr}
\tilde{x}_{4m + 1} & \tilde{x}_{4m + 2} & \tilde{x}_{4m + 3} & \tilde{x}_{4m + 4}\\
- \tilde{x}_{4m + 2}^{*}  & \tilde{x}_{4m + 1}^{*} & - \tilde{x}_{4m + 4}^{*} & \tilde{x}_{4m + 3}^{*}\\
\tilde{x}_{4m + 3} &  \tilde{x}_{4m + 4} & \tilde{x}_{4m + 1} & \tilde{x}_{4m + 2}\\
- \tilde{x}_{4m + 4}^{*}  & \tilde{x}_{4m + 3}^{*} & - \tilde{x}_{4m + 2}^{*} & \tilde{x}_{4m + 1}^{*}\\
\end{array}\right]
\end{equation*}
}
where
$$\tilde{x}_{4m + 1} =  {x}_{(4m + 1)I} + j {x}_{(4m + 4)Q};~ \\
\tilde{x}_{4m + 2} = {x}_{(4m + 2)I} + j {x}_{(4m + 3)Q};$$
$$\tilde{x}_{4m + 3} = {x}_{(4m + 1)Q}  + j {x}_{(4m + 4)I};~ \\
 \tilde{x}_{4m + 4} = {x}_{(4m + 2)Q} + j {x}_{(4m + 3)I}.$$
Step (ii) : Let $\textbf{H}, \mathbf{\Delta}$ and $\mathbf{\Theta} \in \mathbb{C}^{K \times K}$ given by
$\textbf{H} = \mbox{diag}\left\lbrace h_{1}, h_{2}, \cdots, h_{K} \right\rbrace$, 
$\mathbf{\Delta} = \mbox{diag}\left\lbrace 1, 0, 1, 0, \cdots 0\right\rbrace \mbox{ and }
\mathbf{\Theta} = \mbox{diag}\left\lbrace 0, 1, 0, 1, \cdots 1\right\rbrace$
where $h_{1}, h_{2}, \cdots h_{K}$ are complex variables and $\mathbf{\Delta}, \mathbf{\Theta}$ are such that $\mathbf{\Delta} + \mathbf{\Theta}  = \textbf{I}_{K}$. Using $\textbf{H}, \mathbf{\Delta}$ and $\mathbf{\Theta}$, construct a diagonal matrix, $\textbf{G}$ as $\textbf{G} = \textbf{H}\mathbf{\Delta} + \textbf{H}^{*}\mathbf{\Theta}.$\\
Step (iii) : Using  $\mathbf{\Omega}_{m}$, construct a $4x \times 4x$ matrix $\textbf{X}_{m}$ given by $\mathbf{\Omega}_{m} \otimes \textbf{I}_{2}^{\otimes (x - 1)}$ for each $m =  0, 1, \cdots y - 1$.\\
Step (iv) : A RS-PDSSDC, $\textbf{X}\left(N, K\right)$ is constructed using $\textbf{X}_{m}$ and $\textbf{G}$ as $\textbf{X}\left(N, K \right)  = \textbf{G}\left[ \textbf{X}_{0} ~  \textbf{X}_{1} ~ \cdots ~ \textbf{X}_{y - 1} \right]$
where the matrix $\left[ \textbf{X}_{0} ~  \textbf{X}_{2} ~ \cdots ~ \textbf{X}_{y - 1} \right]$ is obtained by juxtaposing the matrices $\textbf{X}_{0}, \textbf{X}_{1}, \cdots, \textbf{X}_{y - 1}$.
\begin{example}
For $N = 4$ and $K = 4$, we have $x = y = 1$. Following Step (i) to Step (iv) in the above construction, we have $\textbf{G} = \mbox{diag}\left\lbrace h_{1}, h_{2}^{*}, h_{3}, h_{4}^{*}\right\rbrace$ and $\textbf{X}_{0} = \mathbf{\Omega}_{0}$. Hence $\textbf{X}\left(4, 4\right)$ is given by,
{\small
\begin{equation}
\label{pdssd_4_relay}
\textbf{X}\left(4, 4\right)  = \left[\begin{array}{rrrr}
h_{1}\tilde{x}_{1} & h_{1}\tilde{x}_{2} & h_{1}\tilde{x}_{3} & h_{1}\tilde{x}_{4}\\
- h_{2}^{*}\tilde{x}_{2}^{*}  & h_{2}^{*}\tilde{x}_{1}^{*} & - h_{2}^{*}\tilde{x}_{4}^{*} & h_{2}^{*}\tilde{x}_{3}^{*}\\
h_{3}\tilde{x}_{3} & h_{3}\tilde{x}_{4} & h_{3}\tilde{x}_{1} & h_{3}\tilde{x}_{2}\\
- h_{4}^{*}\tilde{x}_{4}^{*}  & h_{4}^{*}\tilde{x}_{3}^{*} & - h_{4}^{*}\tilde{x}_{2}^{*} & h_{4}^{*}\tilde{x}_{1}^{*}\\
\end{array}\right].
\end{equation}
}

\noindent where $\tilde{x}_{1} =  {x}_{1I} + j {x}_{4Q};~\tilde{x}_{2} = {x}_{2I} + j {x}_{3Q};~\tilde{x}_{3} = {x}_{1Q}  + j {x}_{4I} \mbox{ and }~\tilde{x}_{4} = {x}_{2Q} + j {x}_{3I}.$
The variables $\tilde{x}_{1}, \tilde{x}_{2}, \cdots \tilde{x}_{4}$ are obtained using the precoding matrices $\textbf{P}$ and $\textbf{Q}$ as given in \eqref{cord_int}. The precoding matrices $\textbf{P}$ and $\textbf{Q}$ are constructed as in Subsection \ref{cont_p_q}. The relay specific matrices $\textbf{A}_{k}, \textbf{B}_{k}$ for the RS-PDSSDC in \eqref{pdssd_4_relay} are as given below,
{\small
\begin{equation*}
\textbf{A}_{1} = \left[\begin{array}{rrrr}
1 & 0 & 0 & 0\\
0  & 1 & 0 & 0\\
0 & 0 & 1 & 0\\
0 & 0 & 0 & 1\\
\end{array}\right] ;~ \textbf{B}_{2} = \left[\begin{array}{rrrr}
0 & 1 & 0 & 0\\
-1  & 0 & 0 & 0\\
0 & 0 & 0 & 1\\
0 & 0 & -1 & 0\\
\end{array}\right];~ \textbf{A}_{3} = \left[\begin{array}{rrrr}
0 & 0 & 1 & 0\\
0  & 0 & 0 & 1\\
1 & 0 & 0 & 0\\
0 & 1 & 0 & 0\\
\end{array}\right]~\mbox{ and }~\textbf{B}_{4} = \left[\begin{array}{rrrr}
0 & 0 & 0 & 1\\
0  & 0 & -1 & 0\\
0 & 1 & 0 & 0\\
-1 & 0 & 0 & 0\\
\end{array}\right].
\end{equation*}
}
$$\textbf{B}_{1} = \textbf{A}_{2} = \textbf{B}_{3} = \textbf{A}_{4} = \textbf{0}_{4}.$$
\end{example}
\begin{example}
For $N$ = 4 and $K$ = 8, we have $y = 1$ and $x = 2$. Following the construction procedure in Case \ref{case1}, $\textbf{X}_{0} = \mathbf{\Omega}_{0} \otimes \textbf{I}_{2} \mbox{ and } \textbf{X}\left(4, 8 \right)  = \textbf{G}\textbf{X}_{0}$. Therefore, $\textbf{X}\left(4, 8 \right)$ is given by.\\
{\small
\begin{equation*}
\textbf{X}\left(4, 8\right)  = \left[\begin{array}{rrrrrrrr}
h_{1}\tilde{x}_{1} & h_{1}\tilde{x}_{2} & h_{1}\tilde{x}_{3} & h_{1}\tilde{x}_{4} & 0 & 0 & 0 & 0 \\
- h_{2}^{*}\tilde{x}_{2}^{*}  & h_{2}^{*}\tilde{x}_{1}^{*} & - h_{2}^{*}\tilde{x}_{4}^{*} & h_{2}^{*}\tilde{x}_{3}^{*} & 0 & 0 & 0 & 0 \\
h_{3}\tilde{x}_{3} &  h_{3}\tilde{x}_{4} & h_{3}\tilde{x}_{1} & h_{3}\tilde{x}_{2} &  0 & 0 & 0 & 0 \\
- h_{4}^{*}\tilde{x}_{4}^{*}  & h_{4}^{*}\tilde{x}_{3}^{*} & - h_{4}^{*}\tilde{x}_{2}^{*} & h_{4}^{*}\tilde{x}_{1}^{*} &  0 & 0 & 0 & 0 \\
0 & 0 & 0 & 0 & h_{5}\tilde{x}_{1} & h_{5}\tilde{x}_{2} & h_{5}\tilde{x}_{3} & h_{5}\tilde{x}_{4}\\
0 & 0 & 0 & 0 & - h_{6}^{*}\tilde{x}_{2}^{*}  & h_{6}^{*}\tilde{x}_{1}^{*} & - h_{6}^{*}\tilde{x}_{4}^{*} & h_{6}^{*}\tilde{x}_{3}^{*}\\
0 & 0 & 0 & 0 & h_{7}\tilde{x}_{3} &  h_{7}\tilde{x}_{4} & h_{7}\tilde{x}_{1} & h_{7}\tilde{x}_{2}\\
0 & 0 & 0 & 0 & - h_{8}^{*}\tilde{x}_{4}^{*}  & h_{8}^{*}\tilde{x}_{3}^{*} & - h_{8}^{*}\tilde{x}_{2}^{*} & h_{8}^{*}\tilde{x}_{1}^{*}
\end{array}\right].\\
\end{equation*}
}\\
\end{example}
\end{case}
\begin{case}\label{case2} $N = 4y$ and $K = 4x + a$ for $a = 1, 2 \mbox{ and } 3$ : In this case, a RS-PDSSDC is constructed in two steps as given below.\\
Step(i) : Construct a RS-PDSSDC for parameters $N = 4y$ and $K = 4(x + 1)$ as given in Case \ref{case1}.\\
Step(ii) : Drop the last $4 - a$ rows of the RS-PDSSDC constructed in Step (i).\\
\begin{example}
When $N$ = 4 and $K$ = 6, the parameters $a$, $x$ and $y$ are 2, 1 and 1 respectively. As given in Case \ref{case2}, a RS-PDSSDC for $N$ = 4 and $K$ = 8 is constructed and the last 2 rows of the design are dropped. The code $\textbf{X}(4, 6)$ is as given below.
{\small
\begin{equation*}
\textbf{X}\left(4, 6 \right)  = \left[\begin{array}{rrrrrrrr}
h_{1}\tilde{x}_{1} & h_{1}\tilde{x}_{2} & h_{1}\tilde{x}_{3} & h_{1}\tilde{x}_{4} & 0 & 0 & 0 & 0 \\
- h_{2}^{*}\tilde{x}_{2}^{*}  & h_{2}^{*}\tilde{x}_{1}^{*} & - h_{2}^{*}\tilde{x}_{4}^{*} & h_{2}^{*}\tilde{x}_{3}^{*} & 0 & 0 & 0 & 0 \\
h_{3}\tilde{x}_{3} &  h_{3}\tilde{x}_{4} & h_{3}\tilde{x}_{1} & h_{3}\tilde{x}_{2} &  0 & 0 & 0 & 0 \\
- h_{4}^{*}\tilde{x}_{4}^{*}  & h_{4}^{*}\tilde{x}_{3}^{*} & - h_{4}^{*}\tilde{x}_{2}^{*} & h_{4}^{*}\tilde{x}_{1}^{*} &  0 & 0 & 0 & 0 \\
0 & 0 & 0 & 0 & h_{5}\tilde{x}_{1} & h_{5}\tilde{x}_{2} & h_{5}\tilde{x}_{3} & h_{5}\tilde{x}_{4}\\
0 & 0 & 0 & 0 & - h_{6}^{*}\tilde{x}_{2}^{*}  & h_{6}^{*}\tilde{x}_{1}^{*} & - h_{6}^{*}\tilde{x}_{4}^{*} & h_{6}^{*}\tilde{x}_{3}^{*}\\
\end{array}\right].
\end{equation*}
}\\
\end{example}
\end{case}
\begin{case}\label{case3} $N = 4y + b$ and $K = 4x + a$ where $b = 1, 2, 3$ and $a = 0, 1, 2, 3$ : In this case, RS-PDSSDCs are constructed in the following 3 steps.\\
Step (i) : Construct a RS-PDSSDC, $\textbf{X}\left(4y, 4x + a \right)$ for parameters $N = 4y$ and $K = 4x + a$ as in Case \ref{case2} using the first $4y$ variables.\\ 
Step (ii) : Construct a DOSTBC, $\textbf{X}^{\prime}\left(b, 4x + a \right)$ with parameters $N = b$ and $K = 4x + a$ using the last $b$ variables as in \cite{ZhK}.\\
Step (iii) : The RS-PDSSDC, $\textbf{X}\left(N, K \right)$ is given by juxtaposing $\textbf{X}\left(4y, 4x + a \right)$ and $\textbf{X}^{\prime}\left(b, 4x + a \right)$ as shown below,
$$\textbf{X}\left(N, K \right) = \left[ \textbf{X}\left(4y, 4x + a \right) ~ \textbf{X}^{\prime}\left(b, 4x + a \right) \right].$$
\begin{example}
When $N = 6$ and $K = 8$, the parameters $b, a , x$ and $y$ are respectively given by 2, 0, 2 and 1.\\
As in Step (i), construct $\textbf{X}\left(4, 8\right)$ as explained in Case \ref{case1} which is given below,
{\small
\begin{equation}
\label{x(6,8)}
\textbf{X}\left(4, 8\right)  = \left[\begin{array}{rrrrrrrrrrrrrrrrr}
h_{1}\tilde{x}_{1} & h_{1}\tilde{x}_{2} & h_{1}\tilde{x}_{3} & h_{1}\tilde{x}_{4} & 0 & 0 & 0 & 0\\
- h_{2}^{*}\tilde{x}_{2}^{*}  & h_{2}^{*}\tilde{x}_{1}^{*} & - h_{2}^{*}\tilde{x}_{4}^{*} & h_{2}^{*}\tilde{x}_{3}^{*} & 0 & 0 & 0 & 0\\
h_{3}\tilde{x}_{3} &  h_{3}\tilde{x}_{4} & h_{3}\tilde{x}_{1} & h_{3}\tilde{x}_{2} & 0 & 0 & 0 & 0\\
- h_{4}^{*}\tilde{x}_{4}^{*}  & h_{4}^{*}\tilde{x}_{3}^{*} & - h_{4}^{*}\tilde{x}_{2}^{*} & h_{4}^{*}\tilde{x}_{1}^{*} &  0 & 0 & 0 & 0\\
0 & 0 & 0 & 0 & h_{5}\tilde{x}_{1} & h_{5}\tilde{x}_{2} & h_{5}\tilde{x}_{3} & h_{5}\tilde{x}_{4}\\
0 & 0 & 0 & 0 & - h_{6}^{*}\tilde{x}_{2}^{*}  & h_{6}^{*}\tilde{x}_{1}^{*} & -h_{6}^{*} \tilde{x}_{4}^{*} & h_{6}^{*}\tilde{x}_{3}^{*}\\
0 & 0 & 0 & 0 & h_{7}\tilde{x}_{3} &  h_{7}\tilde{x}_{4} & h_{7}\tilde{x}_{1} & h_{7}\tilde{x}_{2}\\
0 & 0 & 0 & 0 & - h_{8}^{*}\tilde{x}_{4}^{*}  & h_{8}^{*}\tilde{x}_{3}^{*} & - h_{8}^{*}\tilde{x}_{2}^{*} & h_{8}^{*}\tilde{x}_{1}^{*}\\
\end{array}\right].
\end{equation}
}
According to Step (ii), construct a DOSTBC \cite{ZhK}, $\textbf{X}^{\prime}\left(2, 8 \right)$ as shown below, 
{\small
\begin{equation}
\label{x'(2,8)}
\textbf{X}^{\prime}\left(2, 8 \right)  = \left[\begin{array}{rrrrrrrrrrrrrrrrr}
h_{1}\tilde{x}_{5} & h_{1}\tilde{x}_{6} & 0 & 0 & 0 & 0 & 0 & 0\\
-h_{2}^{*}\tilde{x}_{6}^{*} & h_{2}^{*}\tilde{x}_{5}^{*} & 0 & 0 & 0 & 0 & 0 & 0\\
0 & 0 & h_{3}\tilde{x}_{5} & h_{3}\tilde{x}_{6} & 0 & 0 & 0 & 0\\
0 & 0 & -h_{4}^{*}\tilde{x}_{6}^{*} & h_{4}^{*}\tilde{x}_{5}^{*} & 0 & 0 & 0 & 0\\
0 & 0 & 0 & 0 & h_{5}\tilde{x}_{5} & h_{5}\tilde{x}_{6}& 0 & 0\\
0 & 0 & 0 & 0 & h_{6}^{*}\tilde{x}_{6}^{*} & h_{6}^{*}\tilde{x}_{5}^{*} & 0 & 0\\
0 & 0 & 0 & 0 & 0 & 0 & h_{7}\tilde{x}_{5} & h_{7}\tilde{x}_{6}\\
0 & 0 & 0 & 0 & 0 & 0 & -h_{8}^{*}\tilde{x}_{6}^{*} & h_{8}^{*}\tilde{x}_{5}^{*}
\end{array}\right].
\end{equation}
}
\noindent A RS-PDSSDC $\textbf{X}\left(6, 8 \right)$ is constructed by juxtaposing the designs in \eqref{x(6,8)} and \eqref{x'(2,8)} as shown below,
\begin{equation*}
\textbf{X}\left(6, 8 \right)  = \left[ \textbf{X}\left(4, 8 \right)  ~ \textbf{X}^{\prime}\left(2, 8\right) \right].\\
\end{equation*}
\end{example}
\end{case}
\subsection{Comparison of the Symbol-rates of RS-PDSSDCs and row-monomial DOSTBCs}
\indent For a given value of $N, K$ such that $N \geq 4$ and $K \geq 4$, we proposed a method of constructing a RS-PDSSDC, $\textbf{X}\left(N, K\right)$ with a minimum value of $T$. The minimum values of $T$ provided in our construction is listed below against the corresponding values of $N$ and $K$. Against every value of $T$ for RS-PDSSDCs, the corresponding value of $T$ for row monomial DOSTBC is provided with in the braces.\\
(i) $N$ even, $K$ even :
{\small
\begin{equation*}
\label{K_even_N_even}
T ~ \geq  ~ \left\{ \begin{array}{cccccc}
4xy ~\left(8xy\right) & \mbox{if}& N = 4y, K = 4x.\\
4xy + 4x ~\left(8xy + 4x\right) & \mbox{if}& N = 4y + 2, K = 4x.\\
4xy + 4y ~\left(8xy + 4y\right) & \mbox{if}& N = 4y, K = 4x + 2.\\
4xy + 4y + 4x + 2 ~\left(8xy + 4y + 4x + 2\right)& \mbox{if}& N = 4y + 2, K = 4x + 2.
\end{array}
\right.
\end{equation*}
(ii) $N$ even, $K$ odd :
\begin{equation*}
\label{K_odd_N_even}
T ~ \geq  ~ \left\{ \begin{array}{cccccc}
4xy + 4y ~\left(8xy + 4y\right) & \mbox{if}& N = 4y, K = 4x + 1.\\
4xy + 4y + 4x + 2 ~\left(8xy + 4x + 4y + 2\right) & \mbox{if}& N = 4y + 2, K = 4x + 1.\\
4xy + 4y ~\left(8xy + 8y\right) & \mbox{if}& N = 4y, K = 4x + 3.\\
4xy + 4y + 4x + 4  ~\left(8xy + 8y + 4x + 4\right) & \mbox{if}& N = 4y + 2, K = 4x + 3.
\end{array}
\right.
\end{equation*}
(iii) $N$ odd, $K$ even :
\begin{equation*}
\label{K_even_N_odd}
T ~ \geq  ~ \left\{ \begin{array}{cccccc}
4xy + 4x ~\left(8xy + 4x\right) & \mbox{if}& N = 4y + 1, K = 4x.\\
4xy + 4y + 4x + 2 ~\left(8xy + 4x + 4y + 2\right) & \mbox{if}& N = 4y + 1, K = 4x + 2.\\
4xy + 8x ~\left(8xy + 8x\right) & \mbox{if}& N = 4y + 3, K = 4x.\\
4xy + 4y + 8x + 4  ~\left(8xy + 4y + 8x + 4\right) & \mbox{if}& N = 4y + 3, K = 4x + 2.
\end{array}
\right.
\end{equation*}
(iv) $N$ odd, $K$ odd :
\begin{equation*}
\label{K_odd_N_odd}
T ~ \geq  ~ \left\{ \begin{array}{cccccc}
4xy + 4y + 4x + 1 \\
\left(\mbox{max} \left(8xy + 4x + 2y + 1,~8xy + 4y + 2x + 1\right)\right) & \mbox{if}& N = 4y + 1, K = 4x + 1.\\
4xy + 4y + 4x + 3 \\
\left(\mbox{max} \left( 8xy + 6y + 4x + 3,~8xy + 8y + 2x + 2\right) \right) & \mbox{if}& N = 4y + 1, K = 4x + 3.\\
4xy + 8x + 4y + 3 \\
\left(\mbox{max} \left( 8xy + 6x + 4y + 3,~8xy + 8x + 2y + 2\right) \right) & \mbox{if}& N = 4y + 3, K = 4x + 1.\\
4xy + 4y + 8x + 8 \\
\left(\mbox{max} \left( 8xy + 8x + 6y + 6,~8xy + 8y + 6x + 6\right) \right) & \mbox{if}& N = 4y + 3, K = 4x + 3.
\end{array}
\right.
\end{equation*}
}
\indent From the above comparison, it can be observed that, for a given value of $N$ and $K$, a RS-PDSSDC, $\textbf{X}(N, K)$ is constructed with a smaller value of $T$ compared to that of a row monomial DOSTBC, there by providing higher values of the symbol- rate, $\frac{N}{T}$. In particular, when $N$ is a multiple of 4 and $K$ is of the form 0 or 3 modulo 4, row monomial DOSTBCs need double the number of channel uses in the second phase compared to that of RS-PDSSDCs. It can also be observed that improvement in the values of $T$ for a RS-PDSSDC is not significant when $K$ and $N$ are both odd.
\section{On the construction of S-PDSSDCs with out precoding at the source}\label{sec5}
The existence of high rate S-PDSSDCs has been shown in the preceding sections, when the source performs co-ordinate interleaving of information symbols before broadcasting it to all the relays. In this setup, the relays do not perform coordinate interleaving of their received symbols. One obvious question that needs to be answered is, whether linear processing of the received symbols at the relays alone is sufficient to construct S-PDSSDCs when the source doesn't perform coordinate interleaving of information symbols. In other words, is coordinate interleaving of the information symbols at the source necessary to construct PDSSDCs$?$. The answer is, yes.\\
\indent In the rest of this section, we show that PDSSDCs cannot be constructed by linear processing of the received symbols at the relays when the source transmits the information symbols to all the relays with out precoding.
Towards that end, let the $k^{th}$ relay be equipped with a pair of matrices, $\textbf{A}_{k}$ and $\textbf{B}_{K} \in \mathbb{C}^{N \times T}$ which perform linear processing on the received vector. Excluding the additive noise component, the received vector at the $k^{th}$ relay is 
$h_{k}\textbf{s} = \left[h_{k}x_{1}, h_{k}x_{2}, \cdots h_{k}x_{N} \right]$ where $x_{i}$'s are information symbols and $h_{k}$ is any complex number. The matrices $\textbf{A}_{k}$, $\textbf{B}_{K} \in \mathbb{C}^{N \times T}$ act on the vector $h_{k}\textbf{s}$ to generate a vector of the form,  
\begin{equation}
\label{k_row}
h_{k}\textbf{s}\textbf{A}_{k} + h_{k}^{*}\textbf{s}^{*}\textbf{B}_{k}
\end{equation}
From \eqref{k_row}, the non zero entries of $h_{k}\textbf{s}\textbf{A}_{k} + h_{k}^{*}\textbf{s}^{*}\textbf{B}_{k}$ contains complex variables of the form, $\pm~ x, \pm~ x^{*}$ or multiples of these by $j$ where $j = \sqrt{-1}$ and 
\begin{equation}
\label{real_variables}
\mbox{Re}(x), \mbox{Im}(x) \in \left\lbrace \mbox{Re}\left(h_{k}x_{n}\right) , \mbox{Im}\left(h_{k}x_{n}\right) ~|~ 1 \leq n \leq N \right\rbrace.
\end{equation}
To be precise, $\mbox{Re}\left(h_{k}x_{n}\right)$ and $\mbox{Im}\left(h_{k}x_{n}\right)$ are given by $h_{kI}x_{nI} - h_{kQ}x_{nQ}$ and $h_{kI}x_{nQ} + h_{kQ}x_{nI}$ respectively. The above vector can also contain linear combination of the specified above complex variables.\\
\indent From Definition \ref{D_pdssd}, non-zero entries of the $k^{th}$ row of a PDSSDCs are of the form $\pm~ h_{k}\tilde{x}_{n}$, $\pm~ h_{k}^{*}\tilde{x}_{n}^{*}$ where $\tilde{x}_{nI}$ and $\tilde{x}_{nQ}$ can be in-phase and quadrature components of two different information variables. Since $h_{k}$ is any complex variable, from \eqref{real_variables}, linear processing of the received symbols at the relays alone cannot contribute variables of the form $\pm~ h_{k}\tilde{x}_{n}$, $\pm~ h_{k}^{*}\tilde{x}_{n}^{*}$. Therefore, S-PDSSDCs cannot be constructed by linear processing of the received symbols at the relays alone when the source transmits the information symbols to all the relays with out precoding.
\begin{remark} 
If $h_{k}$'s are real variables, then $\mbox{Re}(x), \mbox{Im}(x) \in \left\lbrace h_{k}\mbox{Re}\left(x_{n}\right) , h_{k}\mbox{Im}\left(x_{n}\right) ~|~ 1 \leq n \leq N \right\rbrace$ in which case, the non-zero entries of the $k^{th}$ row can be of the form $\pm~ h_{k}\tilde{x}_{n}$, $\pm~ h_{k}\tilde{x}_{n}^{*}$ where $\tilde{x}_{nI}$ and $\tilde{x}_{nQ}$ can be in-phase and quadrature components of two different information variables for any real variable $h_{k}$. This aspect has been well studied in \cite{RaR1}, \cite{SCS} and \cite{ZhK2} where the relays are assumed to have the knowledge of phase component of their corresponding channels thereby making  $h_{k}$, a real variable. Hence, with the knowledge of partial CSI at the relays, high rate distributed SSD codes can be constructed by linear processing at the relays alone.  i.e, with the knowledge of partial CSI at the relays, the source need not perform precoding of information symbols before transmitting to the all the relays in the first phase.
\end{remark}
\section{On the full diversity of RS-PDSSDCs}\label{sec6}
\indent In this section, we consider the problem of designing a two-dimensional signal set, $\Lambda$ such that a RS-PDSSDC with variables $x_{1}$, $x_{2}, \cdots x_{N}$ taking values from $\Lambda$ is fully diverse. Since every codeword of a RS-PDSSDC (Definition \ref{D_pdssd}) contains complex variables $h_{k}$'s, Pairwise error probability (PEP) analysis of RS-PDSSDCs is not straightforward. The authors do not have conditions on the choice of a complex signal set such that a RS-PDSSDC is fully diverse. However, we make the following conjecture.\\
\textbf{Conjecture} : A RS-PDSSDC in variables $x_{1}$, $x_{2}, \cdots x_{N}$ is fully diverse if the variables takes values from a complex signal set say, $\Lambda$ such that the difference signal set $\Delta \Lambda$ given by
\begin{equation*}
\Delta \Lambda = \left\lbrace a - b~|~a, b \in \Lambda \right\rbrace 
\end{equation*}
does not have any point on the lines that are $\pm~ 45$ degrees in the complex plane apart from the origin.\\
\indent In the rest of this section, we provide simulation results on the performance comparison of a RS-PDSSDC, $\textbf{X}\left(4, 4 \right)$ (given in \eqref{pdssd_4_relay}) and a row-monomial DOSTBC, $\textbf{X}^{\prime}\left(4, 4 \right)$ (given in \eqref{dostbc_4_relay}) in terms of Symbol Error Rate (SER) (SER corresponds to errors in decoding a single complex variable). The SER comparison is provided in Figure \ref{compare}. Since the design in \eqref{pdssd_4_relay} has double the symbol-rate compared to the design in \eqref{dostbc_4_relay}, for a fair comparison, 16 QAM and a 4 point rotated QPSK are used as signal sets for $\textbf{X}^{\prime}\left(4, 4 \right)$ and $\textbf{X}\left(4, 4 \right)$ respectively to maintain the rate of 1 bits per second per Hertz. The average SNR per channel use for $\textbf{X}^{\prime}\left(4, 4 \right)$ and $\textbf{X}\left(4, 4 \right)$ respectively are $\frac{2p_{1}p_{2}}{p_{1} + 1 + 2p_{2}}$ and $\frac{4p_{1}p_{2}}{p_{1} + 1 + 4p_{2}}$. In order to maintain the same Signal to Noise ratio (SNR), for the design $\textbf{X}^{\prime}\left(4, 4 \right)$, every relay (other than the source) uses twice the power as that for the design $\textbf{X}\left(4, 4 \right)$. The class of DOSTBCs are shown to be fully diverse in \cite{ZhK}. From Figure \ref{compare}, it is observed that $\textbf{X}\left(4, 4 \right)$ provides full diversity, since the SER curve moves parallel to that of $\textbf{X}^{\prime}\left(4, 4 \right)$. It can be noticed from Figure \ref{compare} that the design $\textbf{X}\left(4, 4 \right)$ performs better than $\textbf{X}^{\prime}\left(4, 4 \right)$ by close to 2-3 db.
\begin{equation}
\label{dostbc_4_relay}
\textbf{X}^{\prime}\left(4, 4 \right)  = \left[\begin{array}{rrrrrrrr}
h_{1}x_{1} &  h_{1}x_{2} & h_{1}x_{3} & h_{1}x_{4} &  0 & 0 & 0 & 0 \\
- h_{2}^{*}x_{2}^{*}  & h_{2}^{*}x_{1}^{*} & - h_{2}^{*}x_{4}^{*} & h_{2}^{*}x_{3}^{*} &  0 & 0 & 0 & 0 \\
0 & 0 & 0 & 0 & h_{3}x_{1} &  h_{3}x_{2} & h_{3}x_{3} & h_{3}x_{4}\\
0 & 0 & 0 & 0 & - h_{4}^{*}x_{2}^{*}  & h_{4}^{*}x_{1}^{*} & - h_{4}^{*}x_{4}^{*} & h_{4}^{*}x_{3}^{*}\\
\end{array}\right].
\end{equation}

\section{Conclusion and Discussion}\label{sec7}
\indent We considered the problem of designing high rate, single-symbol decodable DSTBCs when the source is allowed to perform co-ordinate interleaving of information symbols before transmitting it to all the relays. We introduced PDSSDCs (Definition \ref{D_pdssd}) and showed that, DOSTBCs are a special case of PDSSDCs.\\
A special class of PDSSDCs having semi-orthogonal property were defined (Definition \ref{D_so_updssd}). A subset of S-PDSSDCs called RS-PDSSDCs is studied and an upperbound on the maximal rate of such codes is derived. The bounds obtained for RS-PDSSDC are shown to be approximately twice larger than that of DOSTBCs. A systematic construction of RS-PDSSDCs are presented for the case when the number of relays, $K \geq 4$. Codes achieving the bound are found when $K$ is of the form 0 or 3 modulo 4. For the rest of the choices of $K$, S-PDSSDCs meeting the above bound on the rate are not known. The constructed codes are shown to have rate higher than that of row monomial DOSTBCs.\\
Some of the possible directions for future work are as follows:
\begin{itemize}
\item In this paper, we studied a special class of PDSSDCs called Unitary PDSSDCs (See Definition \ref{u_p_dssd}). The design of high rate Non-Unitary PDSSDCs is an interesting direction for future work.
\item The authors are not aware of RS-PDSSDCs achieving the bound on the maximum rate other than the case when $K$ is 0 or 3 modulo 4. The upperbounds on the maximum rate for rest of the values of $K$ possibly can be tightened.
\item A class of S-PDSSDCs was defined, by making every row of the PDSSDC \textbf{R}-non-orthogonal to atmost one of its rows. It will be interesting whether the bounds on the maximal rate of PDSSDCs can be increased further by making a row \textbf{R}-non-orthogonal to more than one of its rows.
\end{itemize}

\begin{center}
LIST OF FIGURES
\end{center}
1.  Wireless relay network.\\
2.  Various class of SSD codes for cooperative networks.\\
3.  Performance comparison of S-PDSSDC and DOSTBC for N = 4 and K = 4 with 1 bps/Hz.\\
\newpage
\begin{figure}
\centering
\includegraphics[width=2.5in]{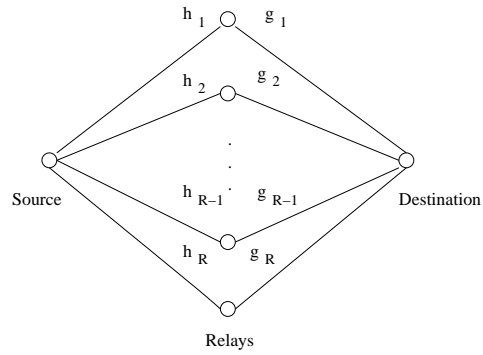}
\caption{Wireless relay network}
\label{model_network}
\end{figure}
\begin{figure}
\centering
\includegraphics[width=2.5in]{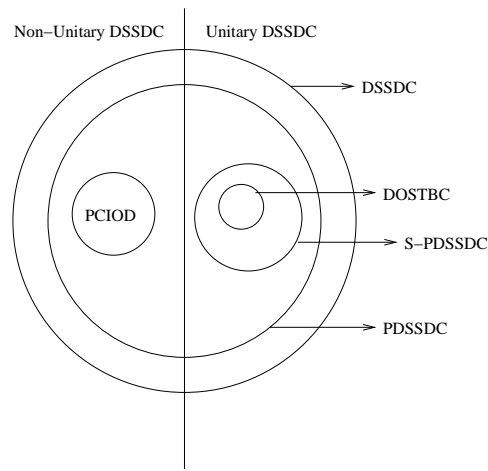}
\caption{Various class of SSD codes for cooperative networks}
\label{various_DSSD}
\end{figure}
\begin{figure}
\centering
\includegraphics[width=2.5in]{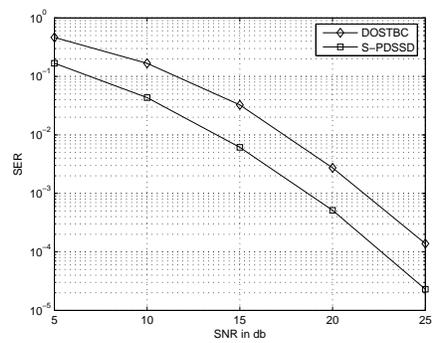}
\caption{Performance comparison of S-PDSSDC and DOSTBC for N = 4 and K = 4 with 1 bps/Hz}
\label{compare}
\end{figure}
\end{document}